\documentclass[11pt]{article}
\usepackage[margin=1in]{geometry}
\usepackage{booktabs}
\usepackage{amsmath, amsfonts, amssymb, amsthm}
\usepackage{thmtools, bm, bbm, thm-restate}
\usepackage{cite}

\usepackage{xcolor}
\definecolor{BrickRed}{rgb}{0.8,0.25,0.33}
\usepackage{hyperref}
\hypersetup{colorlinks,citecolor=blue,linkcolor=BrickRed}
\usepackage[hyperpageref]{backref}
\usepackage{cleveref}
\usepackage{crossreftools}
\usepackage[utf8]{inputenc} % allow utf-8 input
\usepackage[T1]{fontenc}    % use 8-bit T1 fonts
\usepackage{hyperref}       % hyperlinks
\usepackage{url}            % simple URL typesetting
\usepackage{booktabs}       % professional-quality tables
\usepackage{amsfonts}       % blackboard math symbols
\usepackage{nicefrac}       % compact symbols for 1/2, etc.
\usepackage{microtype}      % microtypography
\usepackage{xcolor}         % colors
\usepackage{textcmds}

\usepackage{wrapfig}
\usepackage{graphicx}

\usepackage{algorithm} 
\usepackage{algorithmic}

\usepackage{dsfont}
\usepackage{amsmath,amsthm, amsfonts, amssymb}
\usepackage{cleveref}

\usepackage{mathtools}
\usepackage{thmtools, thm-restate}
\newtheorem{lemma}{Lemma}
\newtheorem{theorem}{Theorem}
\newtheorem*{theorem*}{Theorem}
\newtheorem{corollary}{Corollary}

\theoremstyle{definition}

\newtheorem{example}{Example}
\usepackage{natbib}
\newcommand{\OPT}{\mathrm{OPT}}
\newcommand{\ALG}{\mathrm{ALG}}
\newcommand{\Ex}[2][]{\mbox{\rm\bf E}_{#1}\left[#2\right]}
\renewcommand{\Pr}[2][]{\mbox{\rm\bf Pr}_{#1}\left[#2\right]}
\newcommand{\e}{\mathrm{e}}
\newcommand{\growingmid}{\mathrel{}\middle|\mathrel{}}

\newcommand{\BSF}{\mathrm{BSF}}

\newcommand{\simoverhead}[1][]{\mathrel{\overset{\makebox[0pt]{\mbox{\normalfont\tiny\sffamily #1}}}{\sim}}}

\usepackage{xcolor}
\newcommand{\todo}[1]{}

\title{The Secretary Problem with Predicted Additive Gap}

\author{Alexander Braun\thanks{Institute of Computer Science, University of Bonn. Email: alexander.braun@uni-bonn.de } \qquad Sherry Sarkar\thanks{Carnegie Mellon University. Email: sherrys@andrew.cmu.edu}}

\begin{document}
\maketitle
\begin{abstract}
	The secretary problem is one of the fundamental problems in online decision making; a tight competitive ratio for this problem of $\nicefrac{1}{\e} \approx 0.368$ has been known since the 1960s. Much more recently, the study of algorithms with predictions was introduced: The algorithm is equipped with a (possibly erroneous) additional piece of information upfront which can be used to improve the algorithm's performance. 
Complementing previous work on secretary problems with prior knowledge, we tackle the following question: \\

\emph{What is the weakest piece of information that allows us to break the $\nicefrac{1}{\e}$ barrier?} \\

To this end, we introduce the secretary problem with predicted additive gap. As in the classical problem, weights are fixed by an adversary and elements appear in random order. In contrast to previous variants of predictions, our algorithm only has access to a much weaker piece of information: an \emph{additive gap} $c$. This gap is the difference between the highest and $k$-th highest weight in the sequence.
Unlike previous pieces of advice, knowing an exact additive gap does not make the problem trivial. Our contribution is twofold. First, we show that for any index $k$ and any gap $c$, we can obtain a competitive ratio of $0.4$ when knowing the exact gap (even if we do not know $k$), hence beating the prevalent bound for the classical problem by a constant. Second, a slightly modified version of our algorithm allows to prove standard robustness-consistency properties as well as improved guarantees when knowing a range for the error of the prediction.
\end{abstract}
 \newpage
	\section{Introduction}
\label{Section:Introduction}

The secretary problem is a fundamental problem in online decision making: An adversary fixes non-negative, real-valued weights $w_1 \geq w_2 \geq \dots \geq w_n$ which are revealed online in random order. The decision maker is allowed to accept (at most) one element. At the time of arrival of an element, the decision maker is required to make an immediate and irrevocable acceptance decision. The goal is to maximize the weight of the selected element. A tight guarantee\footnote{In the literature, there are two variants of the secretary problem: \emph{probability-maximization} (maximize the probability of selecting $w_1$) and \emph{value-maximization} (maximize selected weight). 
Throughout the paper, we consider the latter of the two variants.} of $\nicefrac{1}{\e}$ is known since the seminal work of \citet{Lindley1961DynamicPA} and \citet{Dyn63} (also see \citet{10.2307/2245639} or \citet{10.2307/1402748}) and can be achieved with a very simple threshold policy.

In the modern era, the assumption of having no prior information on the weights is highly pessimistic. To go beyond a worst case analysis, researchers have recently considered the setting where we have some sort of learned prediction that our algorithm may use up front. This setting spawned the recent and very successful field of algorithms with predictions. \citet{NEURIPS2020_5a378f84} and \citet{10.1145/3465456.3467623} studied the secretary problem with a prediction of the largest weight in the sequence, and resolve this setting with an algorithm which yields a nice robustness-consistency trade-off. \citet{FY23} consider the secretary problem with an even stronger prediction: A prediction for every weight in the sequence. 

However, predicting the largest element or weight can sometimes be difficult or unfavorable. For example in retail, past data may only contain information about prices and not the true values of buyers \citep{DBLP:conf/focs/KleinbergL03, DBLP:conf/soda/LemeSTW23}; or for data privacy reasons (see e.g.\ \citet{DBLP:conf/colt/AsiFKT23}), only surrogates for largest weights are revealed in history data. 
This motivates to advance our understanding of the following question: \\

\emph{What is the weakest piece of information we can predict that still allows us to break the $\nicefrac{1}{e}$ barrier?} \\

Stated another way, is there a different parameter we can predict, one that does not require us to learn the best value, but is still strong enough to beat $\nicefrac{1}{e}$?  
This brings us to the idea of predicting the \textit{gap} between the highest and $k$-th highest weight, or in other words, predicting how valuable $w_k$ is with respect to $w_1$. Coming back to data privacy for example, such a parameter does successfully anonymize the largest weight in the sequence.

From a theoretical perspective, for some special cases of gaps, previous work directly implies improved algorithms. For example, if we know $w_1$ and $w_2$ have the same weight, using an algorithm for the two-best secretary problem \citep{10.2307/2283044} directly leads to a better guarantee. More generally, if we know the \textit{multiplicative gap} $\nicefrac{w_1}{w_2}$, we observe that we can generalize the optimal algorithm for $w_1 = w_2$ (see e.g. \citet{10.2307/2283044, DBLP:journals/mor/BuchbinderJS14}). However, if we instead only know $\nicefrac{w_n}{w_1}= 0$, this does not help at all. The only insight is that $w_n = 0$; we have no insight on the range of values taken. This essentially boils down to the classical secretary problem and the best competitive ratio again is $\nicefrac{1}{\e}$. 
So while the multiplicative gap advises about the relative values without needing to know $w_1$ entirely, it is not strong enough in general to beat $\nicefrac{1}{\e}$. In this paper, we consider instead predicting an \textit{additive gap} $w_1 - w_k$.

The additive gap between $w_1$ and $w_k$ can be viewed as interpolating between two previously studied setups: when $w_1 - w_k$ gets small, we get closer towards the $k$-best secretary problem, and when $w_1 - w_k$ is very large, the additive gap acts as a surrogate prediction of $w_1$, the prediction setting in \citet{NEURIPS2020_5a378f84} and \citet{10.1145/3465456.3467623}. As we will see, even though the additive gap is much weaker than a direct prediction for $w_1$, it strikes the perfect middle ground: it is strong enough to beat $\nicefrac{1}{e}$ by a constant for any possible value of the gap $w_1 - w_k$ (and even if we do not know what $k$ is upfront). In addition, in contrast to pieces of advice studied in the literature so far, knowing an exact additive gap does not make the problem trivial to solve.

\subsection{Our Results and Techniques}
\label{subsection:resultstechniques}

Our contribution is threefold. First, in \Cref{Section:Generalgap}, we show the aforementioned result: knowing an exact additive gap allows us to beat the competitive ratio of $\nicefrac{1}{\e}$ by a constant.

\begin{theorem}[Theorem~\ref{Theorem:general_gap}, simplified form]
    There exists a deterministic online algorithm which achieves an expected weight of $\Ex[]{\ALG} \geq 0.4 \cdot w_1$ given access to a single additive gap $c_k$ for $c_k = w_1 - w_k$ and some $k$.
\end{theorem}

Still, getting an exact gap might be too much to expect. Hence, in \Cref{section:robustness-consistency}, we introduce a slight modification in the algorithm to make it robust with respect to errors in the predicted gap while simultaneously outperforming the prevalent competitive ratio of $\nicefrac{1}{\e}$ by a constant for accurate gaps.

\begin{theorem}[Theorem~\ref{thm:RC}, simplified form] \label{thm:RC_simplified}
    There exists a deterministic online algorithm which uses a predicted additive gap and is simultaneously $(\nicefrac{1}{\e} + O(1))$-consistent and $O(1)$-robust.
\end{theorem}

The previous \Cref{thm:RC_simplified} does not assume any bounds on the error of the predicted additive gap used by our algorithm. In particular, the error of the prediction might be unbounded and our algorithm is still constant competitive. 
However, if we know that the error is bounded, we can do much better. 

\begin{theorem}[Theorem~\ref{Theorem:approx_gap}, simplified form]
    There exists a deterministic online algorithm which achieves an expected weight of $\Ex[]{\ALG} \geq 0.4 \cdot w_1 - 2 \epsilon$ given access to a bound $\epsilon$ on the error of the predicted gap.
\end{theorem}

Our algorithms are inspired by the one for classical secretary, but additionally incorporate the gap: Wait for some time to get a flavor for the weights in the sequence, set a threshold based on the past observations and the gap, pick the first element exceeding the threshold. 

At first glance, this might not sound promising: In cases when the gap is small, incorporating the gap in the threshold does not really affect the best-so-far term. Hence, it may seem that beating $\nicefrac{1}{\e}$ is still hard. However, in these cases, even though the threshold will be dominated by the best-so-far term most of the time, the gap reveals the information that the best value and all other values up to $w_k$ are not too far off. That is, even accepting a weight which is at least $w_k$ ensures a sufficient contribution.

Our analyses use this fact in a foundational way: Either the gap is large in which case we do not consider many elements in the sequence for acceptance at all. Or the gap is small which implies that accepting any of the $k$ highest elements is reasonably good. For each of the cases we derive lower bounds on the weight achieved by the algorithm. 

Since we do not know upfront which case the instance belongs to, we optimize our initial waiting time for the worse of the two cases. In other words, the waiting time cannot be tailored to the respective case but rather needs to be able to deal with both cases simultaneously. This introduces some sort of tension: For instances which have a large gap, we would like the waiting time to be small. By this, we could minimize the loss which we incur by waiting instead of accepting high weighted elements at the beginning of the sequence. For instances which have a small gap, the contribution of the gap to the algorithm's threshold can be negligible. This results in the need of a longer waiting time at the beginning to learn the magnitude of weights reasonably well. We solve this issue by using a waiting time which balances between these two extremes: It is (for most cases) shorter than the waiting time of $\nicefrac{1}{\e}$ from the classical secretary algorithm. Still, it is large enough to gain some information on the instance with reasonable probability.

As a corollary of our main theorem, we show that we can beat the competitive ratio of $\nicefrac{1}{\e}$ even if we only know the gap $w_1 - w_k$ but do not get to know the index $k$. In particular, this proves that even an information like \qq{There is a gap of $c$ in the instance} is helpful to beat $\nicefrac{1}{\e}$, no matter which weights are in the sequence and which value $c$ attains.

Complementing theoretical results, we run simulations in \Cref{Section:Simulations} which strengthen our theoretical findings. First, we show that for instances in which the classical secretary algorithm achieves a nearly tight guarantee of $\nicefrac{1}{\e}$, our algorithm can almost always select the highest weight. 
In addition, we further investigate the robustness-consistency trade-off of our algorithm. In particular, as it will turn out, underestimating is not as much of an issue as overestimating the exact gap. 
    \subsection{Additional Related Work}
\label{Subsection:related_work}

\emph{Implications of related results on the additive gap.} In the two-best secretary problem, we can pick at most one element but win when selecting either the best or second best element. For this problem, the competitive ratio is upper bounded by approximately $0.5736$ \citep{DBLP:journals/mor/BuchbinderJS14, DBLP:conf/soda/ChanCJ15} (the authors provide an algorithm which matches this bound, so the guarantee is tight). As our setting with $w_1 - w_2 = 0$ can be viewed as a special case, this yields a hardness result; the best any algorithm can perform with the exact additive gap provided upfront is approximately $0.5736$.

\emph{A non-exhaustive list of related work on secretary problems.} Since the introduction of the secretary problem in the 1960s, there have been a lot of extensions and generalizations of this problem with beautiful algorithmic ideas to solve them \citep{10.5555/1070432.1070519, DBLP:conf/approx/BabaioffIKK07, 10.1145/3212512, DBLP:journals/mor/FeldmanSZ18, 10.1007/978-3-642-02930-1_42, 10.1145/1993636.1993716, doi:10.1137/15M1033708, 10.1145/2897518.2897540}.
Beyond classical setups, recent work by \citet{doi:10.1137/1.9781611975994.128} and \citet{DBLP:conf/soda/CorreaCFOT21} studies the secretary problem with sampled information upfront. Here, some elements are revealed upfront to the algorithm which then tries to pick the best of the remaining weights. Guarantees are achieved with respect to the best remaining element in the sequence. 
In addition, there are also papers bridging between the secretary problem and the prophet inequality world, e.g. \citet{DBLP:journals/corr/abs-2011-06516} or \citet{10.1145/3328526.3329627} and many more \citep{bradac_et_al:LIPIcs:2020:11717, kesselheim_et_al:LIPIcs:2020:12479, doi:10.1137/1.9781611977073.53}.

\emph{Algorithms with machine learned advice.} In the introduction, we already scratched the surface of the field on algorithms with predictions. Here, the algorithm has access to some machine learned advice upfront and may use this information to adapt decisions. Initiated by the work of \citet{DBLP:journals/jacm/LykourisV21} and \citet{DBLP:conf/nips/PurohitSK18}, there have been many new and interesting results in completely classical problems within the last years, including ski rental \citep{WZ20}, online bipartite matching \citep{LMRX21}, load balancing \citep{SHVB23}, and many more (see e.g. \citet{NEURIPS2021_161c5c5a, Zeynali_Sun_Hajiesmaili_Wierman_2021, NEURIPS2021_25047349}). 
Since this area is developing very fast, we refer the reader to the excellent website \citet{AlgosWithPredicitons} for references of literature. 

As mentioned before, also the secretary problem itself has been studied in this framework. \citet{NEURIPS2020_5a378f84} consider the secretary problem when the machine learned advice predicts the weight of the largest element $w_1$. Their algorithm's performance depends on the error of the prediction as well as some confidence parameter by how much the decision maker trusts the advice. In complementary work, \citet{10.1145/3465456.3467623} give a bigger picture for secretary problems with machine learned advice. Their approach is LP based and can capture a variety of settings. They assume that the prediction is one variable for each weight (e.g. a 0/1-variable indicating if the current element is the best overall or not). \citet{FY23} assume an even stronger prediction: Their algorithm is given a prediction for every weight in the sequence.
In contrast, we go into the opposite direction and deal with a less informative piece of information.

Our work also fits into the body of literature studying weak prediction models, previously studied for e.g.\ paging \citep{DBLP:conf/icml/0001B0FHL0S23}, online search  \citep{DBLP:conf/innovations/000121}, to just mention a few. In these, several different directions for weak prediction models were considered. For example, the setting in scheduling or caching where the number of predictions is smaller than the input size \citep{DBLP:conf/icml/Im0PP22, DBLP:conf/icml/BenomarP24}.

	\section{Preliminaries}
\label{Section:Preliminaries}

In the secretary problem, an adversary fixes $n$ non-negative, real-valued weights, denoted $w_1 \geq w_2 \geq \dots \geq w_n$. For each of the elements, there is an \emph{arrival time}\footnote{Note that this setting is equivalent to drawing a random permutation of the $n$ elements and revealing elements in this order one by one.} $t_i \simoverhead[\textnormal{iid}] \textnormal{Unif}[0,1]$. Weight $w_i$ is revealed at time $t_i$ and we immediately and irrevocably need to decide if we want to accept or reject this element. Overall, we are allowed to accept at most one element with the objective of maximizing the selected weight. We say that an algorithm is \emph{$\alpha$-competitive} or achieves a \emph{competitive ratio} of $\alpha$ if $\Ex[]{\ALG} \geq \alpha \cdot w_1 = \alpha \cdot \max_i w_i
$, where the expectation is taken over the random arrival times of elements (and possible internal randomness of the algorithm). \\

In addition to the random arrival order, we assume to have access to a single prediction $\hat{c}_k$ for one \emph{additive gap} together with its index $k$. The additive gap for some index $2 \leq k \leq n$ is $c_k \coloneqq w_1 - w_k$. We say that an algorithm has access to an \emph{exact} or \emph{accurate} gap if $\hat{c}_k = c_k$ (as in \Cref{Section:Generalgap}). 
When the algorithm gets a predicted additive gap $\hat{c}_k$ which might not be accurate (as in \Cref{section:robustness-consistency} or \Cref{Section:Approxgap}), we say that $\hat{c}_k$ has \emph{error} $\eta = |\hat{c}_k - c_k|$. We call an algorithm \emph{$\rho$-robust} if the algorithm is $\rho$-competitive regardless of error $\eta$ and we say the algorithm is \emph{$\psi$-consistent} if the algorithm is $\psi$-competitive when $\eta = 0$. 
To fix notation, for any time $\tau \in [0,1]$, we denote by $\BSF(\tau)$ (read  \emph{best-so-far} at time $\tau$) the highest weight which did appear up to time $\tau$. In other words, $\BSF(\tau) = \max_{i : t_i \leq \tau} w_i$. Also, when clear from the context, we drop the index $k$ at the gap and only call it $c$ or $\hat{c}$ respectively.

    \section{Knowing an Exact Gap}
\label{Section:Generalgap}

Before diving into the cases where the predicted gap may be inaccurate in \Cref{section:robustness-consistency} and \Cref{Section:Approxgap}, we start with the setup of getting a precise prediction for the gap. That is, we are given the exact gap $c_k = w_1 - w_k$ for some $2 \leq k \leq n$. We assume that we get to know the index $k$ as well as the value of $c_k$, but neither $w_1$ nor $w_k$. 

Our algorithm takes as input the gap $c$ as well as a waiting time $\tau$. 
This gives us the freedom to potentially choose $\tau$ independent of $k$ if required. As a consequence, we could make the algorithm oblivious to the index $k$ of the element to which the gap is revealed. We will use this in Corollary~\ref{corollary:unknown_k}. 

\begin{algorithm}[h]	
\caption{Secretary with Exact Additive Gap}
\label{Algorithm:SAG_general}
\begin{algorithmic}
    \STATE \textbf{Input:} Additive gap $c$, time $\tau \in [0,1]$ \\
    \STATE Before time $\tau$: \hspace*{3mm} Observe weights $w_i$ \\
    \STATE At time $\tau$: \hspace*{9.5mm} Compute $\BSF(\tau) = \max_{i : t_i \leq \tau} w_i$ \\
    \STATE After time $\tau$: \hspace*{5mm} Accept first element with $w_i \geq \max( \BSF(\tau), c )$
\end{algorithmic}
\end{algorithm}	

This algorithm beats the prevalent competitive ratio of $\nicefrac{1}{\e} \approx 0.368$ by a constant.

\begin{theorem} \label{Theorem:general_gap}
    Given any additive gap $c_k = w_1 - w_k$, for $\tau = 1 - \left(\nicefrac{1}{k+1}\right)^{\nicefrac{1}{k}}$, Algorithm~\ref{Algorithm:SAG_general} achieves a competitive ratio of $\max \left( 0.4, \nicefrac{1}{2} \left(\nicefrac{1}{k+1}\right)^{\nicefrac{1}{k}}  \right)$. %\enspace. \]
\end{theorem}

Note that as $k$ tends towards $n$ and both become large, the competitive ratio approaches $\nicefrac{1}{2}$. 

We split the proof of Theorem~\ref{Theorem:general_gap} in the following two lemmas. Each of them gives a suitable bound on the performance of our algorithm for general waiting times $\tau$ in settings when $w_k$ is small or large.

The first lemma gives a lower bound in cases when $w_k$ is small in comparison to $w_1$.

\begin{lemma}\label{lemma:proof_case_1}
    If $w_k < \frac{1}{2} w_1$, then $\Ex[]{\ALG} \geq (1-\tau) \left( \frac{1}{2} + \frac{1}{2(k-1)} \right) \cdot w_1$.
\end{lemma}

The second lemma will be used to give a bound when $w_k$ is large compared to $w_1$.

\begin{lemma}\label{lemma:proof_case_2}
    If $w_k \geq \frac{1}{2} w_1$, then the following two bounds hold:
    \begin{itemize}
        \item[(i)] $\Ex[]{\ALG} \geq \frac{k+1}{2k} \left( 1 - \tau -(1-\tau)^{k+1}  \right) \cdot w_1$ and 
        \item[(ii)] $\Ex[]{\ALG} \geq \left( \frac{3}{2} \tau \ln\left( \frac{1}{\tau} \right) - \frac{1}{2} \tau(1-\tau)  \right) \cdot w_1$ \enspace.
    \end{itemize} 
    As a consequence, $\Ex[]{\ALG}$ is also at least as large as the maximum of the two bounds.
\end{lemma}

We start with a proof of Lemma~\ref{lemma:proof_case_1} when $w_k$ is small compared to $w_1$.

\begin{proof}[Proof of Lemma~\ref{lemma:proof_case_1}.]
    
Let $w_k < \frac{1}{2} w_1$. Observe that in this case, the gap is quite large, as \[ c_k = w_1 - w_k > w_1 - \frac{1}{2} w_1 = \frac{1}{2} w_1 > w_k \geq \dots \geq w_n \enspace. \]  In particular, the gap is large enough such that the algorithm either selects nothing (if $w_1$ arrives before $\tau$) or some element among $w_1, \dots, w_l$ for some $1 \leq l \leq k-1$. To see this, first observe that the gap never overshoots $w_1$ as we are always ensured that $w_1 \geq c_k$. In addition, we only accept elements strictly larger than $w_k$ in this case. Let us define $l$ to be the index of element with $w_l \geq c_k > w_{l+1}$, i.e. the last element which is not excluded from a possible acceptance by $c_k$.

Hence, we can bound 
\begin{align*}
    \Ex[]{\ALG} & = \Pr[]{\textnormal{Best arrives after } \tau} \cdot \Ex[]{\ALG \growingmid \textnormal{Best arrives after } \tau} \\ 
    & \geq (1- \tau) \left( \frac{1}{2} w_1 \frac{l-1}{l} + \frac{1}{l} w_1 \right) \\ 
    & \geq w_1 (1-\tau) \left( \frac{1}{2} + \frac{1}{2(k-1)} \right) \enspace.
\end{align*}
To see why the first inequality holds, note that the probability of an element to arrive after $\tau$ is precisely $1-\tau$. In addition, conditioned on the best element arriving after $\tau$, we are ensured to accept some element among the first $l$ elements. By the random arrival times, we accept the best element $w_1$ in at least a $\frac{1}{l}$-fraction of scenarios. In the remainder, we will select an element of weight at least $\frac{1}{2} w_1$. The second inequality uses $l \leq k-1$. 
\end{proof}

In this case, we actually accept $w_1$ with much higher probability than $\frac{1}{l}$. In particular, observe that we exclude $w_{l+1},\dots,w_n$ by the gap in the threshold, so the problem actually boils down to solving a secretary instance with $l$ elements. However, the bound presented in the proof of Lemma~\ref{lemma:proof_case_1} is sufficient for our purposes.

Next, we turn to the regime when $w_k$ is large and prove the two bounds from Lemma~\ref{lemma:proof_case_2}.

\begin{proof}[Proof of Lemma~\ref{lemma:proof_case_2} (i).]
    
Let $w_k \geq \frac{1}{2} w_1$. In this case, the gap can be quite small. Still, we are guaranteed that selecting any element among $w_2, \dots, w_k$ achieves at least a weight of $\frac{1}{2} w_1$.

We condition on the event $w_i = \BSF(\tau)$ for $i \in \{2, \hdots, n\}$. For any $i \leq k+1$, having element $w_i$ as $\BSF(\tau)$, we will always accept the first element among $w_1, \dots, w_{i-1}$ arriving after $\tau$. To see this, note that none of these elements will be excluded by the gap $c_k$ in the threshold of the algorithm as \[ c_k = w_1 - w_k \leq w_1 - \frac{1}{2} w_1 = \frac{1}{2} w_1 \leq w_k \enspace. \] We first give a bound on the expected weight achieved by the algorithm conditioned on seeing $w_i$ as the $\BSF(\tau)$ for some $2 \leq i \leq k+1$. When seeing $w_i$ as the $\BSF(\tau)$, we select $w_1$ in a $\frac{1}{i-1}$-fraction of scenarios. In addition, any element $w_2,\dots,w_{i-1}$ is ensured to have a weight at least $\frac{1}{2} w_1$. As a consequence, using that we only consider $i \leq k+1$,
\begin{align}\label{inequality:expected_algo_conditioned}
    \Ex[]{ \ALG \growingmid w_i \textnormal{ is } \BSF(\tau) } \geq \frac{1}{i-1} w_1 + \frac{i-2}{i-1} \cdot \frac{1}{2} w_1 \geq \frac{1}{2} \left(1+\frac{1}{k} \right) \cdot w_1  \enspace.
\end{align}
Using this, we can derive the following lower bound where in the first inequality, we use that when $w_1$ is $\BSF(\tau)$, we will select nothing.

\begin{align*}
    \Ex[]{\ALG} & \geq \sum_{i=2}^{k+1} \Pr[]{w_i \textnormal{ is } \BSF(\tau) } \cdot \Ex[]{ \ALG \growingmid w_i \textnormal{ is } \BSF(\tau) } \\ 
    & = \sum_{i=2}^{k+1} \tau (1-\tau)^{i-1} \cdot \Ex[]{ \ALG \growingmid w_i \textnormal{ is } \BSF(\tau) } \\
    & \stackrel{\eqref{inequality:expected_algo_conditioned}}{\geq} %\sum_{i=2}^{k+1} \tau (1-\tau)^{i-1} \cdot \left( \frac{1}{i-1} w_1 + \frac{i-2}{i-1} \cdot \frac{1}{2} w_1 \right) \\    & \geq 
    w_1 \cdot \frac{1}{2} \left(1+\frac{1}{k} \right) \tau \sum_{i=2}^{k+1} (1-\tau)^{i-1} \\
    &= w_1 \cdot \frac{1}{2} \left(1+\frac{1}{k} \right) \tau \left( \frac{1-(1-\tau)^{k+1}}{\tau} -1 \right) \\
    & = w_1 \cdot \frac{1}{2} \left(1+\frac{1}{k} \right) \left( 1 - \tau -(1-\tau)^{k+1}  \right)
\end{align*}
The second equality uses a geometric sum to simplify the expression.
\end{proof}

In addition, we can use an alternate analysis, which is tighter for small $k$. 

\begin{proof}[Proof of Lemma~\ref{lemma:proof_case_2} (ii).]

Let $w_k \geq \frac{1}{2} w_1$. We only consider the probability of selecting the best or second best element. Observe that the second best element satisfies $w_2 \geq w_k \geq \frac{1}{2} w_1$ by the case distinction, no matter for which $k$ we observe the gap. Our goal will be to lower bound the acceptance probabilities of $w_1$ and $w_2$ with the ones from the classical secretary problem. To this end, we first observe that for any element $w_i$, if $w_i \geq c_k$, then \begin{align}\label{Observation:Select_i_larger_than_gap}
    \Pr[]{\textnormal{Algorithm~\ref{Algorithm:SAG_general} selects } w_i} \geq \Pr[]{\textnormal{An algorithm with threshold $\BSF(\tau)$ selects } w_i} \enspace.
\end{align}
To see why this inequality holds, let $w_i$ arrive after time $\tau$. If $w_i$ is the first element to surpass $\BSF(\tau)$, by the hypothesis that $w_i \geq c_k$, $w_i$ is also the first to surpass the threshold used by Algorithm~\ref{Algorithm:SAG_general}.

Using this, we can lower bound the selection probabilities of $w_1$ and $w_2$ by the ones of an algorithm which is just using $\BSF(\tau)$ as a threshold. The next few lines of calculations are folklore for the classical secretary problem. Still, for the sake of completeness, we reprove them here. 

To bound the probability of selecting $w_1$ by an algorithm which uses $\BSF(\tau)$ as a threshold, observe that if $w_1$ arrives before $\tau$, it will be rejected. If it arrives at some time $x \in (\tau,1]$, there are two cases in which we accept $w_1$: Either no other element did arrive before (by the i.i.d. arrival times, this happens with probability $(1-x)^{n-1}$) or the best element before did arrive before $\tau$. The latter happens with probability $\frac{\tau}{x}$. As a consequence, we get

\begin{align*}
    \Pr[]{\textnormal{Algorithm~\ref{Algorithm:SAG_general} selects } w_1} & \stackrel{\eqref{Observation:Select_i_larger_than_gap}}{\geq} \Pr[]{\textnormal{An algorithm with threshold $\BSF(\tau)$ selects } w_1} \\ 
    & = \int_\tau^1 (1-x)^{n-1}  + \left( 1- (1-x)^{n-1} \right) \frac{\tau}{x} \ dx \geq \int_\tau^1 \frac{\tau}{x} \ dx = \tau \ln\left( \frac{1}{\tau} \right) \space.
\end{align*}

Similarly, for $w_2$ we can compute the same integral after conditioning on $w_1$ arriving after time $x$. Observe that the probability of $w_1$ arriving after time $x$ is precisely $1-x$.

\begin{align*}
    \Pr[]{\textnormal{Algorithm~\ref{Algorithm:SAG_general} selects } w_2} & \stackrel{\eqref{Observation:Select_i_larger_than_gap}}{\geq} \Pr[]{\textnormal{An algorithm with threshold $\BSF(\tau)$ selects } w_2} \\ 
    & = \int_\tau^1 (1-x) \left( (1-x)^{n-2}  + \left( 1- (1-x)^{n-2} \right) \frac{\tau}{x} \right) \ dx \\ & \geq \int_\tau^1 (1-x) \frac{\tau}{x} \ dx = \tau \ln\left( \frac{1}{\tau} \right) - \tau(1-\tau) \enspace.
\end{align*}

Using that $w_2 \geq \frac{1}{2} w_1$, we obtain 

\begin{align*}
    \Ex[]{\ALG} & \geq w_1 \tau \ln\left( \frac{1}{\tau} \right) + w_2 \left( \tau \ln\left( \frac{1}{\tau} \right) - \tau(1-\tau) \right) \\ 
    & \geq w_1 \left( \tau \ln\left( \frac{1}{\tau} \right) + \frac{1}{2} \left( \tau \ln\left( \frac{1}{\tau} \right) - \tau(1-\tau) \right) \right) \\ 
    & = w_1 \left( \frac{3}{2} \tau \ln\left( \frac{1}{\tau} \right) - \frac{1}{2} \tau(1-\tau) \right) \enspace.
\end{align*}
\end{proof}

Having these two lemmas, we can now conclude the proof of the main theorem.

\begin{proof}[Proof of Theorem~\ref{Theorem:general_gap}.]

We use the lower bound obtained by Lemma~\ref{lemma:proof_case_1}. From Lemma~\ref{lemma:proof_case_2}, we take the maximum of the two bounds into consideration. Since we do not know upfront to which case our instance belongs, we can only obtain the minimum of the two as a general lower bound on the weight achieved by the algorithm.

As a consequence, we obtain $\Ex[]{\ALG} \geq \alpha \cdot w_1$ for  
\begin{align}\label{inequality:comp_ratio}
    \alpha \coloneqq \min \left(     \frac{(1-\tau) k}{2(k-1)}  ;   \max \left( \frac{k+1}{2k} \left( 1 - \tau -(1-\tau)^{k+1}  \right); \frac{3}{2} \tau \ln\left( \frac{1}{\tau} \right) - \frac{1}{2} \tau(1-\tau)  \right) \right)
\end{align}
which depends on the waiting time $\tau$. Now, we plug in $\tau = 1 - \left(\frac{1}{k+1}\right)^{1/k}$. First, note that we can bound the maximum in $\alpha$ with the first of the two terms. Factoring out a $1-\tau$, we get 
\begin{align*}
\alpha & \geq (1-\tau) \cdot \min \left( \frac{k}{2(k-1)} ; \frac{k+1}{2k} \left( 1 - (1-\tau)^k \right) \right) \\
& = \left(\frac{1}{k+1}\right)^{1/k} \cdot \min  \left( \frac{k}{2(k-1)} ; \frac{k+1}{2k} \cdot \frac{k}{k+1} \right)  \\
& = \frac{1}{2} \left(\frac{1}{k+1}\right)^{1/k}  \enspace.
\end{align*}

To compensate for the poor performance of this bound for small $k$, we can use basic calculus to state the following.

After plugging in our choice of $\tau$ into Expression~\eqref{inequality:comp_ratio}, the first term is minimized for $k = 7$ for a value of at least $0.43$. For $2 \leq k \leq 11$, the last term is always at least $0.404$ and for any $k \geq 12$, the second term exceeds $0.403$. Hence, we always ensure that $\alpha \geq 0.4$. 
\end{proof}

From a high level perspective, the two lemmas give a reasonable bound depending of we either exclude a lot of elements in the algorithm (\Cref{lemma:proof_case_1}) or if the largest $k$ elements ensure a sufficient contribution (\Cref{lemma:proof_case_2}). 

As a corollary of the proof of Theorem~\ref{Theorem:general_gap}, we also get a lower bound on the weight achieved by the algorithm if we are only given the gap, but not the element which obtains this gap. That is, we are given $c_k$ but not the index $k$.

\begin{corollary}\label{corollary:unknown_k}
    If the algorithm only knows $c_k$, but not $k$, setting $\tau = 0.2$ achieves $\Ex[]{\ALG}\geq 0.4 \cdot w_1 $.
\end{corollary}

The proof mainly relies on the fact that the lower bound we obtained in the proof of Theorem~\ref{Theorem:general_gap} holds for any choice $\tau \in [0,1]$. Also, the algorithm itself only uses the gap to contribute to the threshold. The index $k$ is only used to compute $\tau$. As a consequence, when choosing $\tau = 0.2$ independent of $k$, the algorithm is oblivious to the exact value of $k$, but only depends on the gap $c_k$. For this choice of $\tau$, we can show that $\alpha \geq 0.4$.  

\begin{proof}[Proof of Corollary~\ref{corollary:unknown_k}.]
    Note that the lower bound in Expression~\eqref{inequality:comp_ratio} holds for any choice of $\tau \in [0,1]$. Also, the algorithm itself only uses the gap to contribute to the threshold. The index $k$ is only used to compute $\tau$. As a consequence, when choosing $\tau = 0.2$ independent of $k$, the algorithm is oblivious to the exact value of $k$, but only depends on the gap $c_k$. 

    For $\tau =  0.2$, the value of $\alpha$ in Expression~\eqref{inequality:comp_ratio} satisfies
    \begin{align*}
        \alpha & = \min \left( 0.8 \cdot \frac{k}{2(k-1)} ; \max \left( \frac{k+1}{2k} \left( 0.8 - 0.8^{k+1} \right) ; 0.3 \ln(5) - 0.08 \right)   \right) \\ 
        & \geq \min \left( 0.8 \cdot \frac{1}{2} ; \max \left( \frac{k+1}{2k} \left( 0.8 - 0.8^{k+1} \right) ; 0.4 \right)   \right) \\ 
        & \geq 0.4
    \end{align*}
    and hence, we get a competitive ratio of at least $0.4$. 
\end{proof}

As a consequence, very surprisingly, even if we only get to know \emph{some} additive gap $c_k$ and not even the index $k$, we can outperform the prevalent bound of $\nicefrac{1}{\e}$. Also, observe that this is independent of the exact value that $c_k$ attains and holds for any small or large gaps.

As mentioned before, \Cref{Algorithm:SAG_general} is required to get the exact gap as input. In particular, once the gap we use in the algorithm is a tiny bit larger than the actual gap $c_k$, we might end up selecting no element at all. 

\begin{example}
	We get to know the gap to the smallest weight $c_n = w_1 - w_n$ and the smallest weight $w_n$ in the sequence satisfies $w_n = 0$. Let the gap which we use in \Cref{Algorithm:SAG_general} be only some tiny $\delta > 0$ too large. In other words, we use $c_n + \delta$ as a gap in the algorithm. Still, this implies that our threshold $\max(\BSF(\tau), c_n + \delta)$ after the waiting time satisfies \[ \max(\BSF(\tau), c_n + \delta) \geq c_n + \delta = w_1 + \delta > w_1 \geq w_2 \geq \dots \geq w_n \enspace. \] As a consequence, we end up selecting no weight at all and have $\Ex{\ALG} = 0$. 
\end{example}

This naturally motivates the need to introduce more robust deterministic algorithms in this setting. The next \Cref{section:robustness-consistency} will show that a slight modification in the algorithm and its analysis allows to obtain robustness to errors in the predictions while simultaneously outperforming $\nicefrac{1}{\e}$ for accurate gaps.

    \section{Robustness-Consistency Trade-offs}
\label{section:robustness-consistency}

Next, we show how to slightly modify our algorithm in order to still beat $\nicefrac{1}{\e}$ when getting the correct gap as input, but still be constant competitive in case the predicted gap is inaccurate. The modification leads to \Cref{alg:RC}: Initially, we run the same algorithm as before. After some time $1- \gamma$, we will lower our threshold in order to hedge against an incorrect prediction. 
\begin{algorithm}[H]
\caption{Robust-Consistent Algorithm}
\label{alg:RC}
\begin{algorithmic}
    \STATE \textbf{Input:} Predicted gap $\hat{c}$, times $\tau \in [0,1)$, $\gamma \in [0,1-\tau)$ \\
    \STATE Before time $\tau$: \hspace*{33.5mm} Observe weights $w_i$ \\
    \STATE At time $\tau$: \hspace*{40mm} Compute $\BSF(\tau) = \max_{i : t_i \leq \tau} w_i$ \\
    \STATE Between time $\tau $ and time $1 - \gamma$:  \hspace*{3.9mm} Accept first element with $w_i \geq \max( \BSF(\tau), \hat{c} )$
    \STATE After time $1-\gamma$: \hspace*{28.7mm} Accept first element with $w_i \geq \BSF(\tau)$
\end{algorithmic}
\end{algorithm}	

Note that by $\gamma \in [0,1-\tau)$, we ensure that $\tau < 1-\gamma$, i.e.\ the waiting time $\tau$ is not after time $1- \gamma$ and hence, the algorithm is well-defined. Now, we can state the following theorem which gives guarantees on the consistency and the robustness of \Cref{alg:RC}. We will discuss afterwards how to choose $\tau$ and $\gamma$ in order to outperform the classical bound of $\nicefrac{1}{\e}$ by a constant for accurate predictions while ensuring to be constant-robust at the same time.
\begin{theorem} \label{thm:RC}
    Given a prediction $\hat{c}_k$ for the additive gap $c_k$, define 
    \vspace{-0.25cm}
    \begin{itemize}
        \item $\alpha_1 \coloneqq 1-\gamma -\tau + \tau \ln \left( \frac{1}{1-\gamma} \right)$ and \\ $\alpha_2 \coloneqq \frac{1}{2} \left( (1+\gamma)(1-\tau -\gamma) + \tau \ln\left( \frac{1}{\tau} \right)  +  \tau \ln\left( \frac{1}{1-\gamma} \right)  \right) $,
        \item $\alpha_3 \coloneqq \frac{k+1}{2k} \left( 1 - \tau -(1-\tau)^{k+1}  \right)$ and $\alpha_4 \coloneqq  \frac{3}{2} \tau \ln\left( \frac{1}{\tau} \right) - \frac{1}{2} \tau(1-\tau) $. 
    \end{itemize}
    Then, \Cref{alg:RC} is (i) $\min \left( \min \left( \alpha_1, \alpha_2 \right) , \max \left( \alpha_3 , \alpha_4 \right) \right)$-consistent and (ii) $\left( \tau \cdot \ln \left( \frac{1}{1-\gamma} \right) \right)$-robust. 
\end{theorem}

Observe that if we do not trust the prediction at all, we could set $\tau = \nicefrac{1}{\e}$ and $1-\gamma = \nicefrac{1}{\e}$. Doing so, we do not use the prediction in our algorithm. Still, for this choices, we recover the guarantee from classical secretary of $\nicefrac{1}{\e}$. 
In other words, we can interpret $\gamma$ as a trust parameter for the prediction which also mirrors our risk appetite. If we do not trust the prediction at all or if we are highly risk averse, we can set $1 - \gamma \approx \tau$. If we are willing to suffer a lot in case of an inaccurate prediction (or if we have high trust in the prediction), we will set $\gamma \approx 0$.

\Cref{thm:RC} yields a trade-off between robustness and consistency. In particular, for a fixed level of robustness, we can choose the optimal values for $\tau$ and $1-\gamma$ for the bounds in \Cref{thm:RC} to obtain the plot in \Cref{fig:optimal_RC_tradeoff_plot}. Observe that when not focusing on robustness (i.e. choosing robustness being equal to zero), we can achieve a consistency approximately matching the upper bound of $0.5736$ described in \Cref{Subsection:related_work}. 

\begin{figure}[h]
	\centering
	\includegraphics[width=0.55\textwidth]{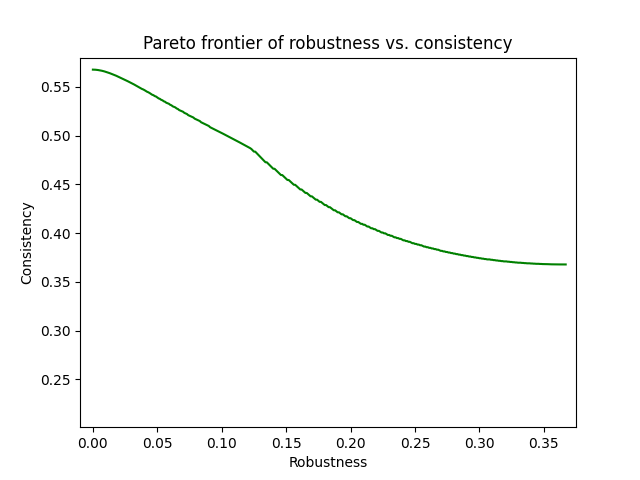}
	\caption{Choosing the optimal parameters $\tau$ and $1- \gamma$ for our analysis in \Cref{thm:RC}: For a given level of robustness, what is the best consistency we can obtain with our analysis.}
	\label{fig:optimal_RC_tradeoff_plot}
\end{figure}

Next, we focus on the proof of \Cref{thm:RC} which we split into two parts: first, we argue about the consistency of our algorithm, second, we show that it is also robust. Concerning robustness, we can only obtain a reasonable contribution by accepting the best weight. Therefore, we derive a lower bound on the probability of accepting the highest weight via \Cref{alg:RC}. Concerning consistency, we can perform a case distinction whether $w_k$ is small or large. Crucially, one is required to take the drop in the threshold after time $1- \gamma$ into account. \\ 

\begin{proof}[Proof of \Cref{thm:RC} (i).]
	For consistency, we assume that our prediction error is zero, hence our predicted gap $\hat{c}_k$ equals the actual gap $c_k$. We can perform the same case distinction as in the proof of \Cref{Theorem:general_gap}. 
	
	\paragraph{Case 1: $w_k < \frac{1}{2} w_1$. }
	
	In the first case, let $w_k < \frac{1}{2} w_1$, hence, \[  \hat{c}_k = c_k = w_1 - w_k > \frac{1}{2} w_1 \geq w_k \geq w_{k+1} \geq \dots \geq w_n \enspace. \] 
	As in the proof of \Cref{lemma:proof_case_1}, we define $l$ to be the index of the element satisfying $w_l \geq c_k > w_{l+1}$, i.e. the last element which is not excluded from a possible acceptance by $c_k$ between time $\tau$ and $1-\gamma$. Note that for any $i \leq l$, we have $w_i \geq w_l \geq \frac{1}{2} w_1$.
	
	Now, we rephrase the expected weight obtained by \Cref{alg:RC} via the expected weight conditioned on the arrival of $w_1$ via
	\begin{align}\label{eqn:split_rc_alg}
		&\Ex{\ALG} \nonumber \\ & = \Pr{w_1 \text{ in } [\tau, 1-\gamma] } \cdot \Ex{\ALG \growingmid w_1 \text{ in } [\tau, 1-\gamma] } + \Pr{w_1 \text{ in } [1-\gamma,1] } \cdot \Ex{\ALG \growingmid w_1 \text{ in } [1-\gamma, 1] } \nonumber \\
		& = (1-\gamma -\tau) \cdot \Ex{\ALG \growingmid w_1 \text{ in } [\tau, 1-\gamma] } + \gamma \cdot \Ex{\ALG \growingmid w_1 \text{ in } [1-\gamma, 1] } \enspace.
	\end{align}
	
	We perform a case distinction if $l=1$ or $l \geq 2$.
	
	If $l = 1$, only $w_1$ exceeds the gap. Hence, we get 
	\begin{align*}
		\Ex{\ALG} & = (1-\gamma -\tau) \cdot \Ex{\ALG \growingmid w_1 \text{ in } [\tau, 1-\gamma] } + \gamma \cdot \Ex{\ALG \growingmid w_1 \text{ in } [1-\gamma, 1] } \\
		& \geq (1-\gamma -\tau) \cdot w_1 + w_1 \cdot \gamma \cdot \int_{1-\gamma}^1 \frac{\tau}{x} \cdot \frac{1}{\gamma} dx \\
		& = w_1 \cdot \left( 1-\gamma -\tau + \tau \cdot \ln \left( \frac{1}{1-\gamma} \right) \right) \eqqcolon w_1 \cdot \alpha_1 \enspace,
	\end{align*}
	where the inequality lower bounds the conditional probability of accepting $w_1$ if we condition on the arrival after time $1- \gamma$. 
	
	If $l \geq 2$, we observe the following. 
	
	In the interval from $\tau$ to $1 - \gamma$, we can pick $w_1$ at some time $x \in [\tau, 1- \gamma]$, if the best up to time $x$ did arrive before $\tau$. Still, in the case where the best up to time $x$ did not arrive before time $\tau$, we are guaranteed to only accept elements which exceed the gap. By the assumption that $c_k > \frac{1}{2} w_1$, we either accept $w_1$ (if the best-so-far is smaller than the gap), or some other element whose value is at least $\frac{1}{2} w_1$. Using a uniform lower bound of $\frac{1}{2} w_1$ in this case, we get  
	\begin{align*}
		&(1-\gamma -\tau) \cdot \Ex{\ALG \growingmid w_1 \text{ in } [\tau, 1-\gamma] } \\ 
		& \geq (1-\gamma-\tau) \cdot \int_{\tau}^{1-\gamma} \frac{1}{1-\tau-\gamma} \cdot \left( \frac{\tau}{x} w_1 + \left(1- \frac{\tau}{x} \right) \frac{1}{2} w_1 \right) dx \\
		& = \frac{1}{2} w_1 \cdot \left( \int_{\tau}^{1-\gamma}  \frac{\tau}{x}  dx + 1 - \tau - \gamma \right) \enspace.
	\end{align*}
	
	For the other term in \Cref{eqn:split_rc_alg}, we use that if $w_1$ arrives after time $1-\gamma$, we will accept a weight of at least $\frac{1}{2} w_1$ if $w_2$ arrives in the interval from $\tau$ to $1-\gamma$. Here it is crucial that $l \geq 2$ since otherwise, the conditioning on the arrival time of $w_1$ would not imply any acceptance in the interval from $\tau$ to $1-\gamma$. Therefore,
	\begin{align*}
		& \gamma \cdot \Ex{\ALG \growingmid w_1 \text{ in } [1-\gamma, 1] } \\ 
		& \geq \gamma \cdot \int_{1-\gamma}^1 \frac{1}{1-(1-\gamma)} \cdot \left( \frac{\tau}{x} w_1 + \frac{1}{2} w_1 \cdot (1-\tau-\gamma) \right) dx \\
		& = \int_{1-\gamma}^1 \frac{\tau}{x} w_1 + \frac{1}{2} w_1 \cdot (1-\tau-\gamma) dx  \\ 
		& = \frac{1}{2} w_1 \cdot \gamma (1-\tau - \gamma) + w_1 \cdot  \int_{1-\gamma}^1 \frac{\tau}{x} dx \enspace,
	\end{align*}
	where the term $(1-\gamma - \tau)$ after the first inequality is the probability of $w_2 $ arriving between $\tau$ and $1-\gamma$. 
	
	Combining this with the lower bound for the first term of \Cref{eqn:split_rc_alg}, we get that for $l \geq 2$, we have
	\begin{align*}
		\Ex{\ALG} & \geq \frac{1}{2} w_1 \cdot \left( \int_{\tau}^{1-\gamma}  \frac{\tau}{x}  dx + 1 - \tau - \gamma \right) +   \frac{1}{2} w_1 \cdot \gamma (1-\tau - \gamma) + w_1 \cdot  \int_{1-\gamma}^1 \frac{\tau}{x} dx \\
		& = \frac{1}{2} w_1 \cdot \int_{\tau}^{1-\gamma}  \frac{\tau}{x}  dx  + \frac{1}{2} w_1 \cdot  (1-\tau-\gamma)+ \frac{1}{2} w_1 \cdot \gamma \cdot  (1-\tau-\gamma) +  w_1 \cdot  \int_{1-\gamma}^1 \frac{\tau}{x} dx \\
		& = \frac{1}{2} w_1 \cdot \left( (1+\gamma)(1-\tau -\gamma) + \int_\tau^1 \frac{\tau}{x} dx  +  \int_{1-\gamma}^1 \frac{\tau}{x} dx \right) \\
		& = \frac{1}{2} w_1 \cdot \left( (1+\gamma)(1-\tau -\gamma) + \tau \cdot \ln\left( \frac{1}{\tau} \right)  +  \tau \cdot \ln\left( \frac{1}{1-\gamma} \right)  \right) \enspace.
	\end{align*}
	Let us define the factor in front of $w_1$ as $$ \alpha_2 \coloneqq \frac{1}{2} \left( (1+\gamma)(1-\tau -\gamma) + \tau \cdot \ln\left( \frac{1}{\tau} \right)  +  \tau \cdot \ln\left( \frac{1}{1-\gamma} \right)  \right) \enspace. $$
	
	Therefore, in case of $w_k < \frac{1}{2} w_1$, we overall get $$ \Ex{\ALG} \geq \min (\alpha_1 , \alpha_2) \cdot w_1 \enspace. $$
	
	\paragraph{Case 2: $w_k \geq \frac{1}{2} w_1$.}
	In the second case, let $w_k \geq \frac{1}{2} w_1$, hence, \[  \hat{c}_k = c_k = w_1 - w_k \leq \frac{1}{2} w_1 \leq w_k \enspace. \] Interestingly, the analysis of \Cref{Algorithm:SAG_general} directly carries over in this case. Recall in this case, we only relied on the gap not excluding high weight elements (with indices $1,\dots,k$), and dropping the gap as a threshold after time $1 - \gamma$ all together preserves this property. When following the proof of \Cref{lemma:proof_case_2} step by step, we can use exactly the same arguments also for \Cref{alg:RC}.
	Hence, we get the same bounds as in \Cref{lemma:proof_case_2} (i) and (ii). 
	
	\paragraph{Combination.} 
	Therefore, defining $\alpha_3 \coloneqq \frac{k+1}{2k} \left( 1 - \tau -(1-\tau)^{k+1}  \right)$ and $\alpha_4 \coloneqq  \frac{3}{2} \tau \ln\left( \frac{1}{\tau} \right) - \frac{1}{2} \tau(1-\tau) $, we obtain  that \[ \Ex[]{\ALG} \geq \alpha \cdot w_1 \] for $\alpha = \min \left( \min \left( \alpha_1, \alpha_2 \right) , \max \left( \alpha_3 , \alpha_4 \right) \right) $.  
\end{proof}

Having shown that our algorithm fulfills the desired consistency properties, we can now shift our perspective towards robustness.

\begin{proof}[Proof of \Cref{thm:RC} (ii).]
	For robustness, we need to protect our algorithm against inaccurate gaps -- no matter how bad the predicted gap is. In particular, once we have access to highly inaccurate gaps, we do not have any information how valuable $w_2, w_3,\dots$ are compared to $w_1$. As a consequence, we bound the expected weight achieved by our algorithm via $$ \Ex{\ALG} \geq \Pr{\text{Select } w_1} \cdot w_1 $$ and aim to find a suitable lower bound on the probability term.
	
	To this end, observe that if the predicted gap is smaller than $w_1$, i.e.\ $\hat{c} \leq w_1$, we can bound the probability term as follows. Denote by $x$ the arrival time of $w_1$. We will always pick $w_1$ if the best element which did arrive before $x$ did arrive before time $\tau$. This happens with probability $\nicefrac{\tau}{x}$, similarly to the proof of \Cref{lemma:proof_case_2}. Hence, if the predicted gap is smaller than $w_1$, we have $$ \Pr[]{\textnormal{Algorithm~\ref{alg:RC} selects } w_1} \geq \int_\tau^1 \frac{\tau}{x} dx = \tau \cdot \ln \left( \frac{1}{\tau} \right) \enspace.$$ 
	
	If the predicted gap is larger than $w_1$, \Cref{alg:RC} will correctly pick $w_1$ in the following case. Again, let $x$ be the arrival time of $w_1$. For any $x$ before time $1- \gamma$, we will not select $w_1$ due to the gap overshooting $w_1$. For $x$ after time $1- \gamma$, we pick $w_1$ at least in the following case: The best element up to time $x$ did arrive before time $\tau$ and hence contributes to the $\BSF(\tau)$-term.
	
	As a consequence, for $\hat{c} > w_1$, we have $$ \Pr[]{\textnormal{Algorithm~\ref{alg:RC} selects } w_1} \geq \int_{1-\gamma}^1 \frac{\tau}{x} dx = \tau \cdot \ln \left( \frac{1}{1-\gamma} \right) \enspace.$$
	
	Therefore, \Cref{alg:RC} will pick $w_1$ with probability at least $\tau \cdot \min \left( \ln \left( \nicefrac{1}{\tau} \right) , \ln \left( \nicefrac{1}{1-\gamma} \right) \right)$. As the function $$ z \mapsto \ln (\nicefrac{1}{z})  $$ is monotonically decreasing and by definition of $\tau$ and $\gamma$, we have $\tau \leq 1-\gamma$, the minimum is attained at $\ln \left( \frac{1}{1-\gamma} \right)$.
	
	As a consequence, \Cref{alg:RC} is $\left( \tau \cdot \ln \left( \frac{1}{1-\gamma} \right) \right)$-robust.
\end{proof}

For example, when using a waiting time $\tau = 0.2$ as in \Cref{corollary:unknown_k} independent of the index $k$ and a value of $\gamma = 0.6$, i.e.\ $1-\gamma = 0.4$, we get the following: \Cref{alg:RC} is approximately $0.383$-consistent and $0.183$-robust (also see \Cref{fig:RC_plot}).

\begin{figure}[h]
	\centering
	\includegraphics[width=0.55\textwidth]{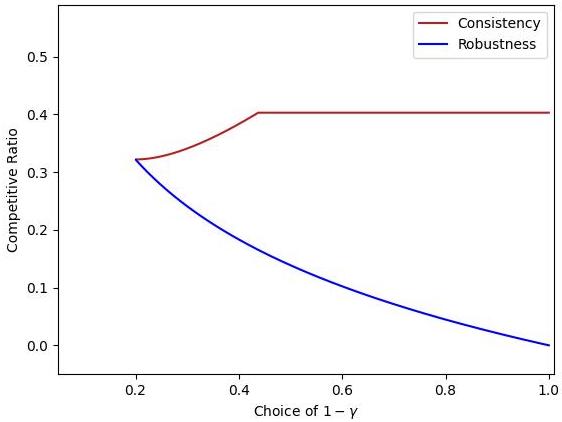}
	\caption{Trade-off between robustness and consistency as a function of the time $1-\gamma$ for fixed choice of $\tau = 0.2$.}
	\label{fig:RC_plot}
\end{figure}

In particular, we can outperform the prevalent bound of $\nicefrac{1}{\e}$ by a constant if the predicted gap is accurate while ensuring to be constant competitive even if our predicted gap is horribly off. Of course, when being more risk averse, one could also increase the robustness guarantee for the cost of decreasing the competitive ratio for consistent predictions.\footnote{We can also slightly improve the trade-off by allowing randomization: Flip a (biased) coin and either choose our algorithm from \Cref{Section:Generalgap} or the classical secretary algorithm which achieves a guarantee of $\nicefrac{1}{\e}$. Our deterministic approach is approximately as good as this randomized variant if we use an unbiased coin and even better compared to biasing the coin towards the algorithm using the prediction. When using a bias towards the classical algorithm, the randomized algorithm has a better trade-off than our deterministic approach.}

We highlight that these guarantees as well as \Cref{thm:RC} hold independent of any bounds on the error of the predicted gap. However, it is reasonable to assume that we have some bounds on how inaccurate our predicted gap is (for example, if our predicted gap is learned from independent random samples). We show in \Cref{Section:Approxgap} that we can achieve much better competitive ratios when we know a range for the error. 
    \section{Improved Guarantees for Bounded Errors}
\label{Section:Approxgap}

Complementing the previous sections where we had either access to the exact gap (\Cref{Section:Generalgap}) or no information on a possible error in the prediction (\Cref{section:robustness-consistency}), we now assume that the error is bounded\footnote{In order to distinguish a bounded error from a possibly unbounded one, we use $\widetilde{c}$ instead of $\hat{c}$ for the predicted gap in this section.}. That is, we get to know some $\widetilde{c}_k \in [c_k - \epsilon; c_k + \epsilon]$ which is ensured to be at most an $\epsilon$ off. Also, the bound $\epsilon$ on the error is revealed to us. Still, the true gap $c_k$ remains unknown.

Our algorithm follows the template which we discussed before. Still, we slightly perturb $\widetilde{c}_k$ to ensure that the threshold is not exceeding $w_1$. 
This algorithm allows to state an approximate version of Theorem~\ref{Theorem:general_gap} for the same lower bounds of $\alpha$ as in the exact gap case. 
\begin{algorithm}[h]
\caption{Secretary with Bounded Prediction Error}
\label{Algorithm:SAG_approx}
\begin{algorithmic}
    \STATE \textbf{Input:} Approximate gap $\widetilde{c}$, time $\tau \in [0,1]$, error bound $\epsilon$ \\
    \STATE Before time $\tau$: \hspace*{3mm} Observe weights $w_i$ \\
    \STATE At time $\tau$: \hspace*{9.5mm} Compute $\BSF(\tau) = \max_{i : t_i \leq \tau} w_i$ \\
    \STATE After time $\tau$: \hspace*{5mm} Accept first element with $w_i \geq \max( \BSF(\tau), \widetilde{c} - \epsilon )$
\end{algorithmic}
\end{algorithm}	

\begin{theorem} \label{Theorem:approx_gap}
    Given any prediction of the gap $\widetilde{c}_k \in [c_k - \epsilon; c_k + \epsilon]$, where $c_k = w_1 - w_k$, Algorithm~\ref{Algorithm:SAG_approx} satisfies $\Ex[]{\ALG}\geq \alpha  \cdot w_1 - 2\epsilon $. 
    For $\tau = 1 - \left(\frac{1}{k+1}\right)^{1/k}$, $\alpha \geq \max \left( 0.4, \frac{1}{2} \left(\frac{1}{k+1}\right)^{1/k}  \right)$ and for $\tau = 0.2$, $\alpha \geq 0.4$.
\end{theorem}
As a consequence, the guarantees from the exact gap case in \Cref{Section:Generalgap} carry over with an additional loss of $2 \epsilon$. Also, the results when not knowing the index $k$ carry over. In particular, this nicely complements the robustness result from \Cref{thm:RC} as follows: Once we can bound the error in a reasonable range, even not knowing the gap exactly does not cause too much of an issue. The proof of \Cref{Theorem:approx_gap} is a straightforward generalization of the proof from \Cref{Theorem:general_gap} and can be found in \Cref{appendix:proof_approx_gap}.

    \section{Simulations}
\label{Section:Simulations}

In order to gain a more fine-grained understanding of the underlying habits, we run experiments\footnote{All experiments were implemented in Python 3.9 and executed on a machine with Apple M1 and 8 GB Memory. } with simulated weights and compare our algorithms among each other and to the classical secretary algorithm\footnote{As the piece of information we use as a prediction is fairly different to the pieces which were used in the literature before, we will not compare our algorithms to other algorithms from the secretary problem with predictions literature.}.

In \Cref{subsection:experiment_gap}, we compare our \Cref{Algorithm:SAG_general} to the classical secretary algorithm. To this end, we draw weights i.i.d. from distributions and execute our algorithm and the classical one. 
As it will turn out, instances which are hard in the normal secretary setting (i.e. when not knowing any additive gap) become significantly easier with additive gap; we can select the best candidate with a much higher probability. We also demonstrate that for some instances, knowing the gap has a smaller impact, though our \Cref{Algorithm:SAG_general} still outperforms the classical one.

Second, in \Cref{subsection:experiment_RC}, we turn towards inaccurate gaps and compare \Cref{Algorithm:SAG_general} developed in \Cref{Section:Generalgap} to the robust and consistent variant of \Cref{alg:RC} from \Cref{section:robustness-consistency}. As a matter of fact, we will see that underestimating the exact gap is not as much of an issue as an overestimation. In particular, underestimating the gap implies a smooth decay in the competitive ratio while overestimating can immediately lead to a huge drop. 

\subsection{The Impact of Knowing the Gap}
\label{subsection:experiment_gap}
We compare our algorithm with additive gap to the classical secretary algorithm (see e.g. \citep{Dyn63}) with a waiting time of $\nicefrac{1}{\e}$. 

\subsubsection{Experimental Setup}

We run the comparison on three different classes of instances:
\begin{itemize}
    \item[(i)] \label{experiment:setup1} \emph{Pareto}: We first draw some $\theta \sim \textnormal{Pareto}(\nicefrac{5}{n}, 1)$. Afterwards, each weight $w_i$ is determined as follows: Draw $Y_i \sim \textnormal{Unif}[0,\theta]$ i.i.d.\ and set $w_i = Y_i^{ \left(n^{1.5} \right)}$ (for more details on Pareto distributions and secretary problems, see e.g. \citet{10.2307/2245639}).
    \item[(ii)] \emph{Exponential}: Here, all $w_i \sim \textnormal{Exp}(1)$. 
    \item[(iii)] \emph{Chi-Squared}: Draw $w_i \sim \chi^2(10)$. That is, each $w_i$ is drawn from a chi-squared distribution which sums over ten squared i.i.d. standard normal random variables.
\end{itemize}
For each class of instances, we average over $5000$ iterations. In each iteration, we draw $n = 200$ weights i.i.d. from the respective distribution together with $200$ arrival times which are drawn i.i.d.\ from $\textnormal{Unif}[0,1]$.
The benchmark is the classical secretary algorithm with a waiting time of $\tau = \nicefrac{1}{\e}$: Set the largest weight up to time $\tau$ as a threshold and accepts the first element afterwards exceeding this threshold. \Cref{Algorithm:SAG_general} is executed with waiting times $\tau = 0.2$ as well as $\tau = 1 - \left(\nicefrac{1}{k+1}\right)^{\nicefrac{1}{k}}$.

\subsubsection{Experimental Results}
\label{sec:simulations_exact}

When weights are sampled based on the procedure explained in (i), we observe an interesting phenomenon (see Figure~\ref{fig:Pareto_plot}). For the classical secretary algorithm, we achieve approximately the tight guarantee of $\nicefrac{1}{\e}$. Our algorithm, however, achieves a competitive ratio of approximately $0.8$ for $\tau = 0.2$. When having a waiting time depending on $k$, we improve the competitive ratio for large $k$ while suffering a worse ratio for small $k$.
\begin{figure}[h]
\centering
    \includegraphics[width=0.55\textwidth]{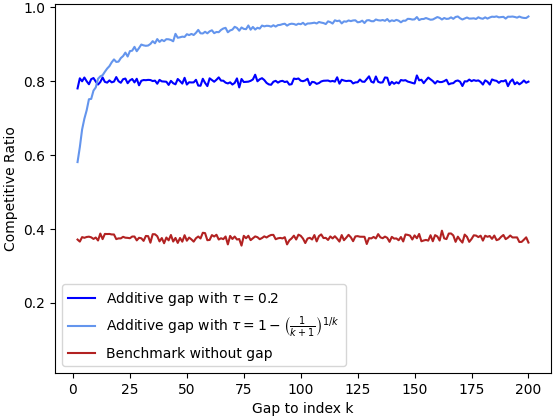}
    \caption{Competitive ratios for weights based on (i). On the $x$-axis, we have the index $k$ from $2$ to $n$. The $y$-axis shows the competitive ratios.}
    \label{fig:Pareto_plot}
\end{figure}
This can be explained as follows. Weights which are distributed according to (i) almost always have a very large gap between the highest and second highest weight. Hence, no matter which gap we observe, it will always be sufficiently large to exclude all elements except the best one. Therefore, we only incur a loss if we do not accept anything (which happens if and only if the best element arrives before the waiting time). As a consequence, for $\tau = 0.2$, we observe the ratio of $0.8$ (which is the probability of the highest weight arriving after time $\tau$). For the waiting times depending on $k$, the waiting time turns out to be larger for smaller $k$ and vice versa. The improvement in the competitive ratio for large $k$ comes from the reduced waiting time and hence a smaller probability of facing an arrival of $w_1$ during the waiting period.

Interestingly, this shows that there are instances for which the classical secretary algorithm almost obtains its tight guarantee of $\nicefrac{1}{\e}$ while these instances become easy when knowing an additive gap. 
As a side remark: One might wonder if it is always true that the index $k$ does not play a pivotal role when using a constant waiting time $\tau = 0.2$. Given the plots below, this is not the case for exponentially distributed weights as in (ii) or Chi-Squared distributed ones as in (iii).

For exponentially distributed weights as in (ii), one can see that even with a static waiting time $\tau = 0.2$, larger indices (and hence automatically larger gaps) are helpful to boost the competitive ratio (see Figure~\ref{fig:Exp_plot}).

\begin{figure}[h]
	\centering
	\includegraphics[width=0.55\textwidth]{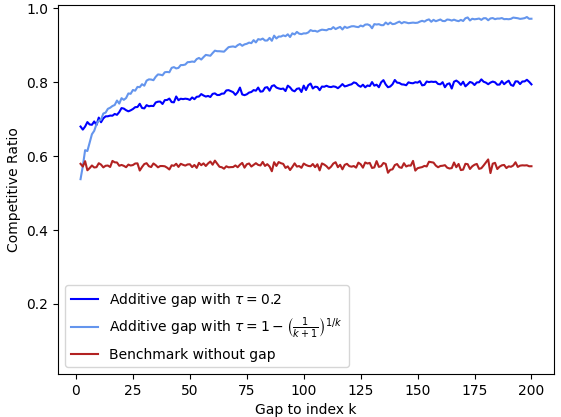}
	\caption{Competitive ratios for weights based on (ii)}
	\label{fig:Exp_plot}
\end{figure}

In addition, for $k = 2$, the waiting time which depends on $k$ does even worse than the classical algorithm. This phenomenon can also be observed for the weights produces by procedure (iii) (see  Figure~\ref{fig:ChiSquared_plot}). It seems that the waiting time for $k=2$ of $1 - 1/\sqrt{3} \approx 0.423$ is simply too large and suffers from losing too much during the exploration phase. Still, also for weights from a Chi-Squared distribution, we can observer that first, knowing gaps helps, and second, larger gaps outperform smaller ones. 

\begin{figure}[h]
	\centering
	\includegraphics[width=0.55\textwidth]{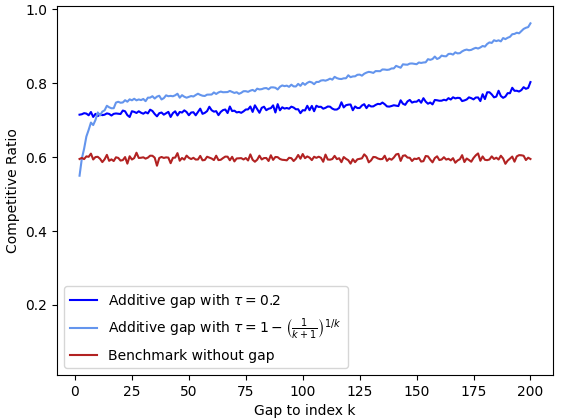}
	\caption{Competitive ratios for weights based on (iii)}
	\label{fig:ChiSquared_plot}
\end{figure}

Summing up, empowering the algorithm with the additional information of some additive gap helps to improve the competitive ratio. As a downside, it turned out that in a few cases it seems that using the index $k$ to compute the waiting time is not beneficial. Still, the waiting time depending on $k$ seems to be an artifact from our analysis. In particular, the waiting time which we used in the simulations was introduced to give provable guarantees. Hence, for practical purposes, one should use a waiting time of $\tau = \min \left( 0.2 ; 1 - \left(\frac{1}{k+1}\right)^{1/k} \right)$ when having access to the index $k$ in order to avoid the waiting time to be too long.

\subsection{Dealing with Inaccurate Gaps}
\label{subsection:experiment_RC}

In order to get a better understanding concerning inaccuracies in the gap, we run a simulation with different errors.

\subsubsection{Experimental Setup}

Again, we average over $5000$ iterations. In each iteration, we set $n = 200$, draw arrival times as before and weights as follows:

\begin{itemize}
    \item[(iv)] \label{experiment:setup2} \emph{Exponential}: Here, all $w_i \sim \textnormal{Exp}(1)$. 
    \item[(v)] \emph{Exponential with superstar}: Here, $w_i \sim \textnormal{Exp}(1)$ for $n-1$ weights and we add a superstar element with weight $100 \cdot \max_{i} w_i $.
\end{itemize}

We compare \Cref{Algorithm:SAG_general} to \Cref{alg:RC} both with waiting time $\tau  = 0.2$. In addition, \Cref{alg:RC} will drop the gap from the threshold after a time of $1 - \gamma = 0.95$, in other words $\gamma = 0.05$. 

The comparison is done for three different gaps: A small one where $k = 2$, i.e.\ the gap between the largest and second largest element, $k = \nicefrac{n}{2}$ and $k=n$, i.e.\ the gap to the smallest element.
Given a multiplication factor $\sigma$ for the error, we feed our algorithm with a predicted gap $\hat{c}_k = \sigma \cdot c_k$ for $\sigma$ going from zero to three in step size of $0.1$. In other words, for $\sigma = 1$, we get an accurate gap, for $\sigma < 1$, we underestimate the gap, for $\sigma > 1 $ we overestimate the gap and for $\sigma = 0$, the algorithms are equivalent to the classical secretary algorithms with waiting time $\tau$.

\subsubsection{Experimental Results}
\label{sec:simulations_RC}

For exponentially distributed weights (see \Cref{fig:exponential_RC}), we can observe that underestimating the gap does not cause too many issues. In particular, when highly underestimating the gap (i.e. $\sigma < 0.5$), both algorithms achieve a competitive ratio of approximately $0.65$, similar to an algorithm not knowing any gap. For an accurate gap, $\sigma = 1$, larger gaps are more helpful as they block more elements from being considered. 
\begin{figure}
\centering
    \includegraphics[width=0.55\textwidth]{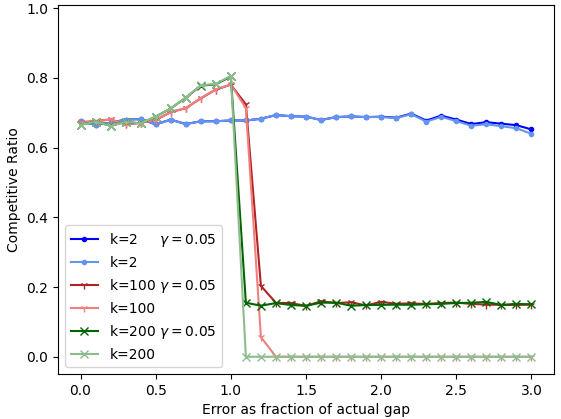}
    \caption{Competitive ratios for weights based on (iv). The $x$-axis shows $\sigma$, where the predicted gap $\hat{c}_k $ used by the algorithms satisfies $\hat{c}_k = \sigma \cdot c_k$ for $\sigma \in [0,3]$.}
    \label{fig:exponential_RC}
\end{figure}
Still, $\sigma > 1$ introduces a transition. For $\sigma > 1$ and gaps between the best and a small element (e.g. $k = 100$ or $k=200$), overestimating the gap reduces the selection probability of \emph{any} weight of \Cref{Algorithm:SAG_general} to zero: The predicted gap is simply too large and even exceeds $w_1$. Still, \Cref{alg:RC} is robust in a sense that we still achieve a competitive ratio of approximately $0.15$. This constant depends on our choice of $\gamma$. As mentioned before, there is the natural trade-off: Increasing $\gamma$ for an improved robustness and suffer a decrease in the competitive ratio for $\sigma =1$. 

Interestingly, for the gap between the best and second best element, both algorithms are much more robust. This can be explained as the gap is small in this case anyway, so overestimating by a factor of three does not cause too much issues yet. One would require to overestimate by a much larger factor here to see a significant difference in the performance of both algorithms.

When considering a more adversarial setting with exponential weights and one superstar element, note that any algorithm can only be constant competitive by selecting the superstar. As illustrated in \Cref{fig:perturbed_exponential_RC}, no matter if we consider the gap to $k=2$, $k=100$ or $k=200$, the gap is always large enough to exclude mainly all elements. 

\begin{figure}[h]
	\centering
	\includegraphics[width=0.55\textwidth]{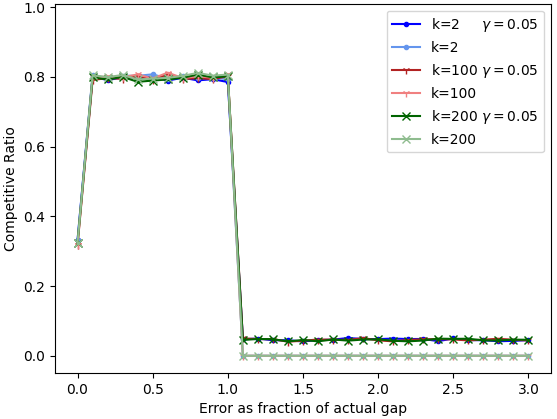}
	\caption{Competitive ratios for weights based on (v). The $x$-axis shows $\sigma$, where the predicted gap $\hat{c}_k = \sigma \cdot c_k$ for $\sigma \in [0,3]$.} 
	\label{fig:perturbed_exponential_RC}
\end{figure}

In addition, even underestimating the gap by a lot (with $\sigma = 0.1$) does not cause any problems. On the other hand, once we overestimate only by a tiny bit, we mainly lose all guarantees and \Cref{Algorithm:SAG_general} becomes not competitive. Our more robust variant in \Cref{alg:RC} achieves a constant competitive ratio which could be increased when choosing larger values of $\gamma$. Again, this would lead to a decrease in the competitive ratio at $\sigma = 1$.

    \section{Conclusion and Future Directions}
\label{Section:conclusion}

As we have seen, a single simple piece of information of the form \qq{There is a gap of $c$ in the instance} helps to improve the competitive ratio for the secretary problem. The reason for this is that we either have a gap which is large, and hence helps to exclude small elements from acceptance. Or, if the gap is small, then we achieve a reasonable fraction of the best weight when accepting the second, third or fourth best. 
In addition, our algorithm can be made robust against inaccurate predictions without sacrificing too much in the competitive ratio. \\

Our results directly impose some open questions for future research.
First, our guarantees seem to be not tight. 
Can we achieve a better competitive ratio for any gap? Or is there a matching hardness result? 
As a second open question, the gaps that we consider are of the form $w_1 - w_k$ for some $k$. As a generalization, one could consider arbitrary gaps $w_i - w_j$ for some $1 \leq i < j \leq n$. Can we do something in this regime? (as sketched in \Cref{appendix:other_gaps}, we can for e.g. $w_2-w_3 = 0$). \\
Also, going beyond the single selection problem is interesting, for example by considering the multi-selection variant. For this, \Cref{section:multi-selection} could be used as a reasonable starting point.

    \paragraph{Acknowledgment}
    The authors would like to thank Thomas Kesselheim for helpful discussions in early stages of this project and the anonymous reviewers for their feedback. This work was done in part while the authors were visiting the Simons Institute for the Theory of Computing for the program on Data-Driven Decision Processes. Alexander Braun has been funded by the Deutsche Forschungsgemeinschaft (DFG, German Research Foundation), Project No. 437739576.

    \newpage    
	\bibliography{references}
    \newpage
    \appendix
    \section{Proof for Gaps with Bounded Error in Theorem~\ref{Theorem:approx_gap}}
\label{appendix:proof_approx_gap}

We give a full proof for Theorem~\ref{Theorem:approx_gap} from Section~\ref{Section:Approxgap}.

\begin{proof}[Proof of Theorem~\ref{Theorem:approx_gap}]
First, we argue that the threshold in the algorithm is never too high to avoid acceptance of $w_1$ if $w_1$ arrives after $\tau$. To see this, note that $\widetilde{c}_k - \epsilon \leq (c_k + \epsilon ) - \epsilon = c_k = w_1 - w_k \leq w_1$. 

For the case distinction, we consider the cases that $w_k$ is small or large with respect to $w_1$. Still, we need incorporate the fact that we know a prediction with bounded error and not the exact gap.

\textbf{Case 1:} $w_k < \frac{1}{2} w_1 - 2\epsilon$.

Observe that in this case, the term $\widetilde{c}_k - \epsilon$ in the threshold is quite large. In particular, 
\begin{align}\label{equation:lower_bound_approx_gap}
    \widetilde{c}_k - \epsilon \geq (c_k -\epsilon) - \epsilon = w_1 - w_k - 2\epsilon > w_1 - (\frac{1}{2} w_1 - 2\epsilon) - 2\epsilon = \frac{1}{2} w_1 \enspace.
\end{align} As mentioned before, $\widetilde{c}_k - \epsilon$ never exceeds $w_1$. Also, observe that $\frac{1}{2} w_1 > \frac{1}{2} w_1 - 2\epsilon > w_k \geq \dots \geq w_n$ by our case distinction.

As a consequence, in this case, the algorithm either selects nothing or some element among $w_1, \dots, w_l$ for some $1 \leq l \leq k-1$. As in the proof in Section~\ref{Section:Generalgap}, we can define $l$ to be the index of element with $w_l \geq \widetilde{c}_k - \epsilon > w_{l+1}$. Also, any element which is selected has a weight of at least $\frac{1}{2} w_1$ by Inequality~\eqref{equation:lower_bound_approx_gap}.

Hence, we achieve the same bound as in the exact gap scenario of
\begin{align*}
    \Ex[]{\ALG} & \geq \Pr[]{\textnormal{Best arrives after } \tau} \cdot \Ex[]{\ALG \growingmid \textnormal{Best arrives after } \tau} \\ 
    & \geq (1- \tau) \left( \frac{1}{2} w_1 \frac{l-1}{l} + \frac{1}{l} w_1 \right) \\ 
    & \geq w_1 (1-\tau) \left( \frac{1}{2} + \frac{1}{2(k-1)} \right) \enspace.
\end{align*}

\textbf{Case 2:} $w_k \geq \frac{1}{2} w_1 - 2\epsilon$.

Observe that selecting any element among $w_2, \dots, w_k$ achieves at least a weight of $\frac{1}{2} w_1 - 2\epsilon$. We now bound the expected weight of the algorithm in a similar way as in Section~\ref{Section:Generalgap} by deriving two lower bounds.

\textit{Bound (i):}

We condition on seeing elements $w_2,\dots, w_{k+1}$ as $\BSF(\tau)$. As before, 

\begin{align*}
    \Ex[]{\ALG} & \geq \sum_{i=2}^{k+1} \Pr[]{w_i \textnormal{ is } \BSF(\tau) } \cdot \Ex[]{ \ALG \growingmid w_i \textnormal{ is } \BSF(\tau) } \\ 
    & = \sum_{i=2}^{k+1} \tau (1-\tau)^{i-1} \cdot \Ex[]{ \ALG \growingmid w_i \textnormal{ is } \BSF(\tau) } \\
    & \geq \sum_{i=2}^{k+1} \tau (1-\tau)^{i-1} \cdot \left( \frac{1}{i-1} w_1 + \frac{i-2}{i-1} \cdot \left( \frac{1}{2} w_1 - 2\epsilon  \right)  \right) \\
    & \geq \left( w_1 \cdot \frac{1}{2} (1+\frac{1}{k}) - 2\epsilon \right) \tau \sum_{i=2}^{k+1} (1-\tau)^{i-1} \\
    & \geq w_1 \cdot \frac{1}{2} (1+\frac{1}{k}) \left( 1 - \tau -(1-\tau)^{k+1}  \right) - 2\epsilon \enspace.
\end{align*}
The only difference to Section~\ref{Section:Generalgap} is the lower bound for $w_2,\dots,w_k$ which are only guaranteed to be at least $\frac{1}{2} w_1 - 2\epsilon$.

\textit{Bound (ii):}

As before, in order to compensate for the weak lower bound in the small $k$ regime, we only consider the probability of selecting the best or second best element. Observe that Inequality~\eqref{Observation:Select_i_larger_than_gap} also holds for \Cref{Algorithm:SAG_approx} holds when replacing the condition on the weights with $w_i \geq \widetilde{c}_k - \epsilon$.

Still, for the following reason, we need to argue in a slightly different way than in the exact gap case. Setting the contribution to the threshold to $\widetilde{c}_k - \epsilon$ ensures that the threshold never overshoots $w_1$. Still, the weight $w_2$ can now fall in two ranges: (a) $w_2 \geq \widetilde{c}_k - \epsilon$ and hence, $w_2$ is not affected by the gap in the threshold or (b) $w_2 < \widetilde{c}_k - \epsilon$ in which case the algorithm does not select $w_2$ as the threshold is too high. The latter case might occur as $w_2 \geq w_k \geq \frac{1}{2}w_1 - 2 \epsilon$, but $\widetilde{c}_k - \epsilon \leq c_k + \epsilon - \epsilon = w_1 - w_k \leq \frac{1}{2} w_1 + 2\epsilon$. Still, this will not introduce any problems.

Concerning (a), we argue as before. The second best element satisfies $w_2 \geq w_k \geq \frac{1}{2} w_1 - 2 \epsilon$ by the case distinction. 
Similar to the proof for the exact gap, we bound the probabilities of selecting $w_1$ or $w_2$ and use the lower bound on $w_2$. This implies that 

\begin{align*}
    \Ex[]{\ALG} & \geq w_1 \tau \ln\left( \frac{1}{\tau} \right) + w_2 \left( \tau \ln\left( \frac{1}{\tau} \right) - \tau(1-\tau) \right) \\ 
    & \geq w_1 \tau \ln\left( \frac{1}{\tau} \right) + \left( \frac{1}{2} w_1 - 2 \epsilon \right) \left( \tau \ln\left( \frac{1}{\tau} \right) - \tau(1-\tau) \right) \\ 
    & \geq w_1 \left( \frac{3}{2} \tau \ln\left( \frac{1}{\tau} \right) - \frac{1}{2} \tau(1-\tau) \right) - 2\epsilon \enspace.
\end{align*}

Concerning (b), note that if $w_2$ is excluded from acceptance by the threshold, so is any $w_i \neq w_1$. Hence, the algorithm will always either select nothing (if $w_1$ appears before $\tau$) or accept $w_1$ (if $w_1$ appears after $\tau$). Here it is important that the contribution of $\widetilde{c}_k - \epsilon$ never exceeds $w_1$. As a consequence, in this case,

\begin{align*}
    \Ex[]{\ALG} = w_1 (1- \tau) \geq w_1 \left( \frac{3}{2} \tau \ln\left( \frac{1}{\tau} \right) - \frac{1}{2} \tau(1-\tau) \right) - 2\epsilon \enspace.
\end{align*}

\textbf{Combination.}

When combining everything, we obtain $\Ex[]{\ALG} \geq \alpha \cdot w_1 - 2 \epsilon$ for 
\begin{align*}
    \alpha \coloneqq \min \left(   (1-\tau)  \frac{k}{2(k-1)}  ;   \max \left( \frac{k+1}{2k} \left( 1 - \tau -(1-\tau)^{k+1}  \right); \frac{3}{2} \tau \ln\left( \frac{1}{\tau} \right) - \frac{1}{2} \tau(1-\tau)  \right) \right)  \enspace,
\end{align*}
which is the same bound as in Expression~\eqref{inequality:comp_ratio} with an additional additive loss of $2\epsilon$.

Plugging in $\tau = 1 - \left(\frac{1}{k+1}\right)^{1/k}$ if knowing the index $k$ or $\tau = 0.2$ when not knowing the gap proves the statement.

\end{proof}

    \newpage
    \section{Informative Examples}

We give some informative examples which should give deeper insights and could partly serve as inspirations and starting points for future research on secretary problems with additive gaps.

\subsection{Knowing that $w_2 - w_3 = 0$}
\label{appendix:other_gaps}

We quickly sketch that also other gaps might be helpful. To this end, assume that $w_1 > w_2 = w_3 > w_4 \geq \dots \geq w_n$. Consider the algorithm that waits for time $\tau$ and accepts the first element after $\tau$ strictly exceeding $\BSF(\tau)$. The algorithm knows $w_2-w_3 = 0$ or equivalently $w_2 = w_3$ and can hence choose $\tau$ accordingly. This allows to compute the following for the probability of selecting the best element.

\begin{align*}
    \Pr[]{\textnormal{Select } w_1} & = \sum_{i=2}^n \Pr[]{w_i \textnormal{ is } \BSF(\tau) } \cdot \Pr[]{\textnormal{Select } w_1 \growingmid w_i \textnormal{ is } \BSF(\tau) } + \Pr[]{\textnormal{No arrival before } \tau} \cdot \frac{1}{n} \\ 
    & = \tau(1-\tau) + \tau(1-\tau)^2 + \sum_{i=4}^n \tau (1-\tau)^{i-1} \frac{1}{i-1} + \frac{1}{n} (1-\tau)^n \\
    & = \frac{1}{2} \tau(1-\tau)^2 + \sum_{i=1}^{n-1} \tau (1-\tau)^{i} \frac{1}{i} + \frac{1}{n}  (1-\tau)^n \\
    & \approx \frac{1}{2} \tau(1-\tau)^2 + \tau \ln \left( \frac{1}{\tau} \right)
\end{align*}

Note that the second equality holds as even if $w_3$ is the best-so-far at time $\tau$, we accept $w_1$ and do not accept $w_2$. The third equality rearranges terms and applies an index shift. Afterwards, we argue that for large $n$, the sum is approximately the expansion of the logarithm. 

Optimizing this, we can set $\tau = 0.359$ to ensure that $\Pr[]{\textnormal{Select } w_1} \geq 0.441 > \frac{1}{\e} $. 

On the other hand, for some gaps, as e.g. $w_{n-1} - w_n$, trying to improve upon the $\nicefrac{1}{\e}$ feels completely hopeless. It would be very interesting to get a more detailed picture here: Which gaps are helpful and how much can we gain? Which gaps are mainly useless? Are there any upper bounds?
    \subsection{Beyond Single Selection}
\label{section:multi-selection}

Instead of having the choice to select a single element, we can shift our perspective and focus on the case where we can select up to $L$ elements. When $L$ turns large, asymptotically we can get a $1 - O (\nicefrac{1}{\sqrt{L}})$ \citep{10.5555/1070432.1070519}. Complementing this, recent work by \citet{DBLP:journals/tcs/AlbersL21} improved the guarantees in the small $L$ regime. We also show that knowing an additive gap can help to beat $\nicefrac{1}{\e}$ for small values of $L$ fairly easily.

Now, the optimum solution will not only pick the largest weight but rather the $L$ largest weights, i.e. \[ \OPT = \sum_{j=1}^L w_j \enspace. \]

In order to come up with a competitive algorithm, we make use of the \emph{virtual algorithm} by \citet{DBLP:conf/approx/BabaioffIKK07} and equip it with the additive gap $c = w_L - w_{L+1}$ between the smallest weight which is included in the optimum solution. 

The algorithm works as follows. We overload notation a bit and denote by $r_L$ an element as well as its weight during the run of the algorithm.

\begin{algorithm}[H]	
\caption{$L$-Selection Secretary with Exact Additive Gap}
\label{Algorithm:Lselection}
\begin{algorithmic}
    \STATE \textbf{Input:} Additive gap $c$, time $\tau \in [0,1]$ \\
    \STATE \hspace*{7mm} \textbf{Before time $\tau$:} \\ Observe weights $w_i$. \\
    \STATE \hspace*{7mm} \textbf{At time $\tau$:} \\ Compute reference set $R(\tau)$ which contains $L$ highest weights seen so far. Denote by $r_L$ the current $L$-th largest element in $R$ (if it exists, otw. $r_L = 0$). \\
    \STATE \hspace*{7mm} \textbf{After time $\tau$:} \\ As $w_i$ arrives, if $w_i \geq \max(r_L,c)$ and $r_L$ did arrive before $\tau$, accept $w_i$, add $w_i$ to $R$, remove $r_L$ from $R$ and update $r_L$.
\end{algorithmic}
\end{algorithm}	

As a consequence, during any point in time, we ensure that $R$ always contains the $L$ highest weights seen so far which exceed the gap. 

Having this, we can state the following theorem.

\begin{theorem}\label{theorem:L_selection}
    Algorithm~\ref{Algorithm:Lselection} guarantees \[  \Ex[]{\ALG} \geq \left( \frac{1}{\e} + \frac{\beta}{2 \e} \left( 1 - \frac{1}{ L} + \frac{1}{L \cdot \e^L}  \right) \right) \cdot \OPT \enspace, \] where $\beta$ is the fraction of $\OPT$ which is covered by $w_L$, i.e. $w_{L} = \beta \cdot \OPT$. 
\end{theorem}

\begin{proof}
    As in the single selection case, we make a case distinction whether the $w_{L+1}$ is large or small with respect to $w_{L}$.

    \textbf{Case 1:} $w_{L+1} < \frac{1}{2} w_L$. \\
    In this case, we have that \[ c = w_L - w_{L+1} > w_L - \frac{1}{2} w_L = \frac{1}{2} w_L > w_{L+1} \enspace. \] As a consequence, after time $\tau$, the algorithm will discard any weight smaller than $w_L$ automatically. In addition, any element which is at least $w_L$ will be selected as long as it arrives after the waiting time $\tau$. 

    Hence, for $1 \leq j \leq L$, we get \[ \Pr[]{\text{select } w_j} = 1- \tau \enspace. \]

    \textbf{Case 2:} $w_{L+1} \geq \frac{1}{2} w_L$. \\
    Now, the gap can be fairly small. In particular, we have \[ c = w_L - w_{L+1} \leq w_L - \frac{1}{2} w_L = \frac{1}{2} w_L \leq w_{L+1}  \enspace. \]

    Still, we are ensured that selecting $w_{L+1}$ instead of $w_L$ leads to a sufficient contribution to the total selected weight. 

    To this end, first note that any $1 \leq j \leq L$ is never excluded by the gap. In addition, we select weight $w_j$ in the following cases: Either less than $L$ elements did arrive before $w_j$ or element $r_L$ in set $R$ at the arrival of weight $w_j$ did arrive before the waiting time $\tau$. Having the arrival of weight $w_j$ at time $x$, the probability of $r_L$ arriving before time $\tau$ is $\nicefrac{\tau}{x}$. Therefore, we get   
    \begin{align*}
        \Pr[]{\text{select } w_j } & = \int_\tau^1 \Pr[]{\text{less than } L \text{ arrivals before time } x} + \Pr[]{\text{at least } L \text{ arrivals before time } x} \frac{\tau}{x} \ dx \\ & 
        \geq \int_\tau^1 \frac{\tau}{x} \ dx = \tau \ln\left( \frac{1}{\tau} \right) \enspace.
    \end{align*}

    In addition, note that for some time $x$, we have 
    \begin{align*}
        \Pr[]{\exists i \in [L] : t_i > x} = 1 - \Pr[]{\forall i \in [L] : t_i \leq x } = 1 - x^L
    \end{align*}
    and therefore, we can compute 
    \begin{align*}
        \Pr[]{\text{select } w_{L+1} } & \geq \int_\tau^1 \Pr[]{\exists i \in [L] : t_i > x} \frac{\tau}{x} dx \\ 
        & = \int_\tau^1 \left(1- x^L \right) \frac{\tau}{x} dx \\ 
        & =  \tau \ln\left( \frac{1}{\tau} \right) - \frac{\tau}{L} \left( 1 - \tau^L \right) \enspace.
    \end{align*}
    As a consequence, using that $w_{L} = \beta \cdot \OPT$, we have
    \begin{align*}
        \Ex[]{\ALG} & \geq \tau \ln\left( \frac{1}{\tau} \right) \sum_{j=1}^L w_j + \left( \tau \ln\left( \frac{1}{\tau} \right) - \frac{\tau}{L} \left( 1 - \tau^L \right) \right) \frac{1}{2} w_L \\ 
        & = \left(  \tau \ln\left( \frac{1}{\tau} \right) + \frac{\beta}{2} \left( \tau \ln\left( \frac{1}{\tau} \right) - \frac{\tau}{L} \left( 1 - \tau^L \right) \right) \right) \cdot \OPT
    \end{align*}

    \textbf{Combination:} As a consequence, when combining the two cases, we get the smaller of the two guarantees as our competitive ratio. 
    Setting $\tau = \nicefrac{1}{\e}$, we obtain a guarantee of 
    \begin{align*}
    \min\left(1-\frac{1}{\e} \ ; \ \frac{1}{\e} + \frac{\beta}{2 \e} \left( 1 - \frac{1}{ L} + \frac{1}{L \cdot \e^L}  \right)   \right) \\ 
    = \frac{1}{\e} + \frac{\beta}{2 \e} \left( 1 - \frac{1}{ L} + \frac{1}{L \cdot \e^L}  \right)   \enspace.
    \end{align*}
    The equality holds as for any $L \geq 2$ and $\beta \in [0,1]$, the minimum is obtained by the second expression. 
\end{proof}
\end{document}